\documentclass[11pt, a4paper]{article}

\newcommand{\lv}[1]{#1}
\newcommand{\sv}[1]{}

\usepackage[utf8]{inputenc}
\usepackage{microtype,ellipsis}
\usepackage[numbers,sort]{natbib}
\usepackage{amssymb, amsthm, amsmath}
\usepackage{mathtools}
\usepackage{paralist}
\usepackage{tabularx}
\usepackage{xspace}
\usepackage[pagebackref,pdfdisplaydoctitle,menucolor=orange!40!black,filecolor=magenta!40!black,urlcolor=blue!40!black,linkcolor=red!40!black,citecolor=green!40!black,colorlinks]{hyperref}
\usepackage{todonotes}
\usepackage{booktabs}
\usepackage{multirow}
\usepackage{dsfont}
\usepackage{a4wide}

\usepackage{authblk}

\usepackage{scrpage2}

\usepackage[sort&compress]{cleveref}

\makeatletter
\def\NAT@spacechar{~}
\makeatother

\newtheorem{theorem}{Theorem}
\newtheorem{observation}{Observation}
\newtheorem{corollary}{Corollary}
\newtheorem{lemma}{Lemma}
\newtheorem{rrule}{Rule}
\newtheorem{proposition}{Proposition}
\theoremstyle{definition}
\newtheorem{definition}{Definition}

\newtheorem*{namedthm}{\namedthmname}
\newcounter{namedthm}



\Crefname{reduction}{Reduction}{Reductions}
\Crefname{rrule}{Rule}{Rules}
\Crefname{corollary}{Corollary}{Corollaries}
\Crefname{observation}{Observation}{Observations}
\Crefname{proposition}{Proposition}{Propositions}
\crefname{cond}{Condition}{Conditions}
\creflabelformat{cond}{#2(#1)#3}
\renewcommand{\P}{\text{\normalfont P}}
\newcommand{\NP}{\text{\normalfont NP}}
\newcommand{\APX}{\text{\normalfont APX}}
\newcommand{\FPT}{\text{\normalfont FPT}}

\newcommand{\PTAS}{\text{\normalfont PTAS}\xspace}

\newcommand{\ie}{{\em i.\,e.}}

\DeclareMathOperator{\opt}{opt}

\newcommand{\NCP}{\textsc{Non-Collinear Packing}\xspace}
\newcommand{\GPS}{\textsc{General Position Subset Selection}\xspace}
\newcommand{\optGPS}{\textsc{Maximum General Position Subset Selection}\xspace}

\newcommand{\IS}{\textsc{Independent Set}\xspace}
\newcommand{\VC}{\textsc{Vertex Cover}\xspace}

\newcommand{\PLC}{\textsc{Point Line Cover}\xspace}
\newcommand{\HS}[1]{$#1$-\textsc{Hitting Set}\xspace}
\newcommand{\R}{\mathds{R}}
\newcommand{\N}{\mathds{N}}
\newcommand{\Q}{\mathds{Q}}

\newcommand{\coNPinNPpoly}{\text{coNP}\subseteq\text{NP/poly}}

\newcommand{\problemdef}[3]{
	\begin{center}
  \begin{minipage}{0.95\textwidth}
    \noindent
    \textsc{#1}

			\vspace{2pt}
			\setlength{\tabcolsep}{3pt}
			\begin{tabularx}{\textwidth}{@{}lX@{}}
					\textbf{Input:} 		& #2 \\
					\textbf{Question:} 	& #3
				\end{tabularx}
  \end{minipage}
	\end{center}
}

\makeatletter
\newcommand\footnoteref[1]{\protected@xdef\@thefnmark{\ref{#1}}\@footnotemark}
\makeatother

\begin{document}

\title{Finding Points in General Position}
\author[1]{Vincent Froese}

\author[2]{Iyad Kanj\footnote{Supported by the DFG project DAPA (NI 369/12) during a Mercator fellowship when staying at TU~Berlin.}}

\author[1]{André Nichterlein}

\author[1]{Rolf Niedermeier}

\affil[1]{Institut f\"ur Softwaretechnik und Theoretische Informatik, TU Berlin, Germany,\newline \texttt{\{vincent.froese, andre.nichterlein, rolf.niedermeier\}@tu-berlin.de}}

\affil[2]{School of Computing, DePaul University, Chicago, USA,\newline \texttt{ikanj@cs.depaul.edu}}

\date{}

\maketitle

\begin{abstract}
	We study the \GPS problem: Given a set of points in the plane, find a maximum-cardinality subset of points in general position.
	We prove that \GPS is NP-hard, APX-hard, and present several fixed-parameter tractability results for the problem as well as a subexponential running time lower bound based on the Exponential Time Hypothesis.
\end{abstract}

\section{Introduction}
For a set~$P=\{p_1,\ldots,p_n\}$ of~$n$ points in the plane, a subset~$S\subseteq P$ is in \emph{general position}
if no three points in~$S$ are \emph{collinear} (that is, lie on the same
line). A frequent assumption for point set problems in computational
geometry is that the given point set is in general position.
In this work, we consider the problem of computing a maximum-cardinality subset of points in general position from a given set of points.
This problem has received quite some attention from the combinatorial geometry perspective, but
it was hardly considered from the computational complexity perspective.
In particular, to the best of our knowledge, the classical complexity of the aforementioned problem until now was unknown.
Formally, the decision version of the problem is as follows:

\noindent \problemdef{\GPS}
{A set~$P$ of points in the plane and $k\in\N$.}
{Is there a subset~$S\subseteq P$ in general position of cardinality at least~$k$?}

A well-known special case of \GPS, referred to as the \textsc{No-Three-In-Line} problem, asks to place
a maximum number of points in general position on an $n\times n$-grid. Since at most two points can be placed on any grid-line, the maximum number of points in general position
that can be placed on an $n \times n$-grid is at most $2n$. Indeed, only for small~$n$ it is
known that $2n$~points can always be placed on the $n\times n$-grid. Erd\H{o}s~\cite{Rot51} observed that, for sufficiently large $n$,
one can place $(1-\epsilon)n$ points in general position on the $n\times n$-grid, for any $\epsilon > 0$. This lower bound was improved by Hall et al.~\cite{HJSW75} to
$(\frac{3}{2}-\epsilon)n$. It was conjectured by Guy and Kelly~\cite{MR68} that, for sufficiently large $n$, one can place more than $\frac{\pi}{\sqrt{3}}n$~points in general position on
an $n\times n$-grid. This conjecture remains unresolved, hinting at the challenging combinatorial nature of \textsc{No-Three-In-Line}, and hence of \GPS as well.

A problem closely related to \GPS is \textsc{Point Line Cover}:
Given a point set in the plane, find a minimum-cardinality set of
lines, the size of which is called the
\emph{line cover number}, that cover all points.
Interestingly, the size of a maximum subset in general position is related to the line cover number  (see \Cref{obs:line}).
While \textsc{Point Line Cover} has been intensively studied,
we aim to fill the gap for \GPS by providing
both computational hardness and fixed-parameter tractability
results for the problem.
In doing so, we particularly consider the parameters
solution size~$k$ (size of the sought subset in general position) and its dual $h:=n-k$, and investigate
their impact on the computational complexity of \GPS{}.

\paragraph*{Related Work}
\citet{PW13} provide lower bounds on the size of a point set in general position, a question originally studied by \citet{Erdos86}.
In his Master's thesis, \citet{Cao12} gives a problem kernel of~$O(k^4)$ points for
\GPS (there called \NCP problem) and a simple greedy
$O(\sqrt{\opt})$-factor approximation algorithm for the maximization version.
He also presents an Integer Linear Program formulation and shows
that it is in fact the dual of an Integer Linear Program
formulation for \PLC.
As to results for the much more studied \PLC, we refer to~\citet{KPR16} and the work cited therein.
Also, the problem of deciding whether any three points of a given point set are collinear has been studied recently~\cite{Barba16}.

\paragraph*{Our Contributions}
We show that \GPS is NP-hard and APX-hard and we prove a subexponential lower bound based on the Exponential Time Hypothesis.
Our main algorithmic results, however, concern the power of polynomial-time
data reduction for \GPS: We give an $O(k^3)$-point problem kernel
and an $O(h^2)$-point problem kernel, and show that the latter kernel is
asymptotically optimal under a reasonable complexity-theoretic assumption. \Cref{tab:results} summarizes our results.
\begin{table*}[t!]
	\caption{Overview of the results we obtain for \GPS, where~$n$ is the number of input points, $k$ is the parameter size of the sought subset in general position, $h=n-k$ is the dual parameter, and~$\ell$ is the line cover number. \vspace*{3mm}}
	\label{tab:results}
		\centering
		\def\arraystretch{1.1}
		\begin{tabular}{cll}
			\toprule
			& Result & Reference\\
			\midrule
			\multirow{4}{*}{\rotatebox[origin=c]{90}{Hardness}} & NP-hard & \Cref{thm:NP_APX_hard}\\
			& APX-hard & \Cref{thm:NP_APX_hard}\\
			& no $2^{o(n)}\cdot n^{O(1)}$-time algorithm (unless the ETH fails.)  & \Cref{thm:NP_APX_hard}\\
			& no $O(h^{2-\epsilon})$-point kernel (unless $\coNPinNPpoly$.) & \Cref{thm:dual_points_lower_bound}\\
			\midrule
			\multirow{6}{*}{\rotatebox[origin=c]{90}{Tractability}} & $(15k^3)$-point kernel (computable in~$O(n^2 \log n)$ time) & \Cref{thm:cubic-kernel}\\
			& $O(n^2\log{n}+41^k \cdot k^{2k})$-time solvable & \Cref{cor:fptalgo}\\
			& $(2h^2 + h)$-point kernel (computable in~$O(n^2)$ time) & \Cref{thm:dual_kernel}\\
			& $O(2.08^h + n^3)$-time solvable & \Cref{prop:fptalgoDualParam}\\
			& $(120\ell^3)$-point kernel (computable in~$O(n^2 \log n)$ time) &  \Cref{cor:line_cover}\\
			& $O(n^2\log{n}+41^{2\ell} \cdot \ell^{4\ell})$-time solvable &  \Cref{cor:line_cover}\\
			\bottomrule
		\end{tabular}

\end{table*}

\section{Preliminaries}
\label{sec:prelim}

In this section we introduce some basic definitions.

\paragraph{Geometry}

All coordinates of points are assumed to be represented by rational numbers (we denote the set of rational numbers by~$\Q$).
The {\em collinearity} of a set of points $P$ is the maximum number of points in $P$ that lie on the same line.
A {\em blocker} for two points $p, q$ is a point on the open line segment $pq$.

\paragraph{Graphs}

Let $G = (V(G), E(G))$ be an undirected graph.
We write $|G|$ for $|V(G)| + |E(G)|$.
A vertex $u\in V(G)$ is a {\em neighbor} of (or is {\em adjacent} to) a vertex $v\in V(G)$ if $\{u,v\} \in E(G)$.
The {\em degree} of a vertex $v$ is the number of its neighbors.

An {\em independent set} of a graph $G$ is a set of vertices such that no two vertices in this set are adjacent.
A {\em maximum independent set} is an independent set of maximum cardinality. A {\em vertex
   cover} of $G$ is a set of vertices such that each edge in $G$ is
 incident to at least one vertex in this set.
The NP-complete {\sc Independent Set} problem is: Given an undirected graph~$G$ and $k \in \N$, decide whether $G$ has an independent set of cardinality~$k$.
The {\sc Maximum Independent Set} problem is the optimization version of {\sc Independent Set}, which is the problem of computing an independent set of maximum cardinality in a given graph.
 The NP-complete {\sc Vertex Cover} problem is: Given an undirected graph $G$ and $k \in \N$,
 decide whether $G$ has a vertex cover of cardinality~$k$.

\paragraph{Parameterized Complexity}
A {\em parameterized problem} is a set of instances of the form $({\cal I}, k)$, where  ${\cal I} \in \Sigma^*$ for a finite alphabet set $\Sigma$, and $k \in \N$ is the {\em parameter}.
A parameterized problem $Q$ is {\it fixed-parameter tractable}, shortly \FPT, if there exists an algorithm that on input  $({\cal I}, k)$ decides whether $({\cal I}, k)$ is a yes-instance of~$Q$ in $f(k)|{\cal I}|^{O(1)}$ time,  where $f$ is a computable function independent of $|{\cal I}|$.
A parameterized problem $Q$ is \emph{kernelizable} if there exists a polynomial-time algorithm that maps an instance $({\cal I},k)$ of $Q$ to another instance $({\cal I}',k')$ of $Q$ such that:
 \begin{compactenum}[(1)]
   \item $|{\cal I}'| \leq \lambda(k)$ for some computable function
     $\lambda$,
   \item $k' \leq \lambda(k)$, and
   \item $({\cal I},k)$ is a yes-instance of $Q$ if and only if
     $({\cal I}',k')$ is a yes-instance of $Q$.
 \end{compactenum}
The instance $({\cal I}',k')$ is called a {\em problem kernel} of $({\cal I}, k)$.
A parameterized problem is \FPT~if and only if it is kernelizable~\cite{CCDF97}.
A general account on applying methods from parameterized complexity analysis to problems from computational geometry is due to \citet{GKW08}.

\lv{
\paragraph{Approximation}

A \emph{polynomial-time approximation scheme} (PTAS) for a maximization problem~$Q$ is an algorithm~$A$ that takes an instance~$I$ and a constant~$\rho > 1$ and returns in $O(n^{f(\rho)})$ time for some computable function~$f$ a solution of value at least~$\opt(I) / \rho$, where~$\opt(I)$ denotes the value of an optimal solution of~$I$.
A maximization problem that is \APX-hard does not admit a \PTAS unless $\P = \NP$.
Refer to the books by~\citet{Vaz01} and \citet{WS11} for a more comprehensive discussion of approximation algorithms.
}

\paragraph{Exponential Time Hypothesis}
The Exponential Time Hypothesis (ETH)~\cite{IPZ01} states that \textsc{3-SAT} cannot be solved in~$2^{o(n)}\cdot n^{O(1)}$ time, where~$n$ is the number of variables in the input formula.

\section{Hardness Results}
\label{sec:hardness}

In this section, we prove that \GPS is \NP-hard, \APX-hard, and presumably not solvable in subexponential time.
Our hardness results follow from a transformation (mapping arbitrary
graphs to point sets) that is based on a construction due to
\citet[Section~5]{GR15}, which they used to prove the NP-hardness of the
\IS problem on so-called \emph{point visibility graphs}.
This transformation, henceforth called~$\Phi$, allows us to obtain the above-mentioned hardness results (using reductions from NP-hard restrictions of \IS to \GPS).
Moreover, in \Cref{subsec:dual}, we will use~$\Phi$ to give a reduction
from \VC to \GPS in order to obtain kernel size
lower bounds with respect to the dual parameter
(see \Cref{thm:dual_size_lower_bound} and~\Cref{thm:dual_points_lower_bound}).
We start by formally defining some properties that are required for the output point set of the transformation.
As a next step, we prove that such a point set can be realized
in polynomial time.

Let $G$ be a graph with vertex set~$V(G)=\{v_1, \dots, v_n\}$.
Let~$C=\{p_1,\ldots,p_n\}$ be a set of points that are in strictly convex position (that is, the points in~$C$ are vertices of a convex polygon), where~$p_i\in C$ corresponds to~$v_i$, $i=1, \dots, n$.
For each edge~$e=\{v_i,v_j\} \in E(G)$, we place a \emph{blocker}~$b_e$ on the line segment~$p_ip_j$ such that the following three conditions are satisfied:
\begin{compactenum}[(I)]
  \item \label[cond]{cond:bpp} For any edge $e \in E(G)$ and for any two points $p_i, p_j \in C$, if $b_e, p_i, p_j$ are collinear, then $p_i, p_j$ are the points in $C$ corresponding to the endpoints of edge $e$.
  \item \label[cond]{cond:bbp} Any two distinct blockers $b_e, b_{e'}$ are not collinear with any point $p_i \in C$.
  \item \label[cond]{cond:bbb} The set~$B:=\{b_e\mid e\in E(G)\}$ of blockers is in general position.
\end{compactenum}

See \Cref{fig:reduction} for an example of the transformation described above.

\begin{figure}[t]
  \centering
  \begin{tikzpicture}
    \tikzstyle{vertex} = [draw, circle, inner sep=0pt];

    \node[vertex] (v1) at (0,2) {$v_1$};
    \node[vertex] (v2) at (2,2) {$v_2$};
    \node[vertex] (v3) at (2,0) {$v_3$};
    \node[vertex] (v4) at (0,0) {$v_4$};
    \node[vertex] (v5) at (-1,1) {$v_5$};

    \draw (v1) -- (v2);
    \draw (v1) -- (v5);
    \draw (v1) -- (v4);
    \draw (v2) -- (v3);
    \draw (v3) -- (v5);
    \draw (v4) -- (v5);
  \end{tikzpicture}
  \hspace{5em}
  \begin{tikzpicture}[scale=0.66]
    \tikzstyle{point} = [draw, circle, inner sep=0pt, minimum size=5pt];
    \tikzstyle{blocker} = [draw=black, circle, fill=black, inner sep=0pt, minimum size=4pt];

    \node[point, label=above:{$p_1$}] (p1) at (0,5) {};
    \node[point, label=above:{$p_2$}] (p2) at (3,3) {};
    \node[point, label=below:{$p_3$}] (p3) at (1.5,0) {};
    \node[point, label=below:{$p_4$}] (p4) at (-1.5,0) {};
    \node[point, label=above:{$p_5$}] (p5) at (-3,3) {};

    \draw[draw=gray] (p1) -- node[pos=0.5,blocker, label=above:{$b_{12}$}] (b12) {} (p2);
    \draw[draw=gray] (p1) -- node[pos=0.5,blocker, label=above:{$b_{15}$}] (b15)  {} (p5);
    \draw[draw=gray] (p1) -- node[pos=0.3,blocker, label=right:{$b_{14}$}] (b14) {} (p4);
    \draw[draw=gray] (p2) -- node[pos=0.5,blocker, label=right:{$b_{23}$}] (b23) {} (p3);
    \draw[draw=gray] (p3) -- node[pos=0.3,blocker, label=above:{$b_{35}$}] (b35) {} (p5);
    \draw[draw=gray] (p4) -- node[pos=0.5,blocker, label=left:{$b_{45}$}] (b45) {} (p5);
  \end{tikzpicture}
  \caption{Example of a graph (left) and a point set (right) satisfying \cref{cond:bpp,cond:bbp,cond:bbb}. The set contains a point (white) for each vertex and a blocker (black) for each edge in the graph such that the only collinear triples are $(p_i,p_j,b_{ij})$ for every edge~$\{v_i,v_j\}$. Also, no four points in the set are collinear.
  }
  \label{fig:reduction}
\end{figure}
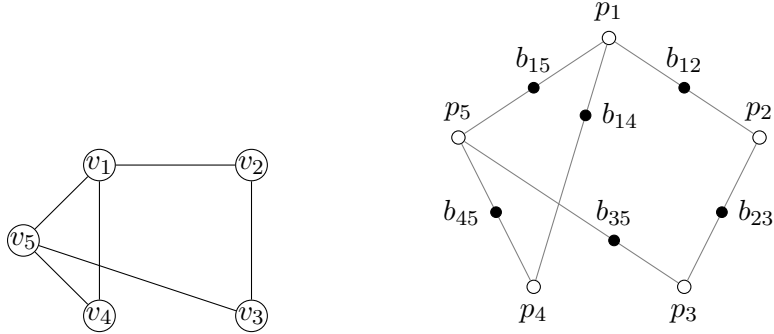

\begin{proposition}\label[proposition]{prop:valid}
  There is a polynomial-time transformation~$\Phi$ mapping arbitrary graphs to point sets that satisfy \cref{cond:bpp,cond:bbp,cond:bbb}. Moreover, no four points in the point set~$C\cup B$ produced by~$\Phi$ are collinear.
\end{proposition}
{
\begin{proof}
Given a graph $G$, let $n=|V(G)|$ and let $C=\{p_1,\ldots,p_n\}$ be a set of rational
points that are in a strictly convex
position; for instance, let $p_j := (\frac{2j}{1+j^2}, \frac{1-j^2}{1+j^2})$ for~$j\in\{1,\ldots,n\}$ be~$n$ rational points on the unit circle centered at the origin \cite{tan96}.
To choose the set~$B$ of blockers,
suppose (inductively) that we have chosen a subset $B'$ of blockers such
that all blockers in~$B'$ are rational points and satisfy
\cref{cond:bpp,cond:bbp,cond:bbb}.
Let~$b_e \not\in B'$ be a blocker corresponding to an edge
$e=\{v_i,v_j\}$ in~$G$.
To determine the coordinates of~$b_e$, we first mark the intersection points (if any) between the line segment $p_ip_j$ and the lines formed by every pair of distinct blockers in~$B'$, every pair of distinct points in~$C \setminus \{p_i, p_j\}$, and every pair consisting of a blocker in~$B'$ and a point in $C \setminus \{p_i, p_j\}$.
We then choose $b_e$
to be an interior point of $p_ip_j$ with rational coordinates that is
distinct from all marked points.
To this end, let~$q$ be the first marked point on the segment~$p_ip_j$ (starting from~$p_i$),
and let~$b_e$ be the midpoint of~$p_iq$. This point is rational since it is the midpoint of rational points.
It is easy to see that~$C\cup B$ can be constructed
in polynomial time and that all points in $C\cup B$ are rational
and satisfy \cref{cond:bpp,cond:bbp,cond:bbb}. Moreover, it easily follows from the construction of $C\cup B$ that no four points in $C \cup B$ are collinear.
\end{proof}
}

\sv{
	Using transformation~$\Phi$, we can prove (proof omitted) the following hardness results via reductions from (variants of) \IS.
}
\lv{
In what follows, we will use transformation~$\Phi$ as a reduction from
\IS to \GPS in order to prove our hardness results.
The following observation will be helpful in proving the correctness
of the reduction.

\begin{observation}
  \label[observation]{obs:solutionstructure}
  Let~$G$ be an arbitrary graph, and let~$P:=\Phi(G)=C\cup B$. For any point set $S\subseteq P$ that is in general position, there is a general position set of size at least~$|S|$ that contains the set of blockers~$B$.
\end{observation}
 
{
\begin{proof}
  Suppose that $S \subseteq P$ is in general position,
  and suppose that there is a point~$b\in B\setminus S$.
  If~$b$ does not lie on a line defined by any two points
  in~$S$, then~$S\cup\{b\}$ is in general position.
  Otherwise, $b$ lies on a line defined by two points $p, q \in S$.
  By \cref{cond:bpp,cond:bbp}, it
  holds that $p,q\in C$. Moreover, $p$ and~$q$ are the only two points
  in~$S$ that are collinear with~$b$.
  Hence, we exchange one of them with~$b$ to obtain
  a set of points in general position of the same cardinality as~$S$.
  Since $b \in B$ was arbitrarily chosen, we can repeat the above
  argument to obtain a subset in general position of cardinality at
  least~$|S|$ that contains~$B$.
\end{proof}
}

Using \Cref{obs:solutionstructure}, we can give a polynomial-time
many-one reduction from \IS to \GPS based on transformation~$\Phi$.

\begin{lemma}\label[lemma]{lem:IS-to-GPS}
  There is a polynomial-time many-one reduction from \IS to \GPS. Moreover, each instance of \GPS produced by this reduction satisfies the property that no four points in the instance are collinear.
\end{lemma}
{
\begin{proof}
  Let~$(G,k)$ be an instance of \IS, where $k \in \N$. The \GPS
  instance is defined as~$(P:=\Phi(G),k+|E(G)|)$.
  Clearly, by \Cref{prop:valid}, the set~$P$ can be computed in
  polynomial time, and no four points in $P$ are collinear.
  We show that~$G$ has an independent
  set of cardinality $k$ if and only if~$P$ has a
  subset in general position of cardinality $k + |E(G)|$.

  Suppose that $I \subseteq V(G)$ is an independent set of cardinality
  $k$, and let~$S:=\{p_i\mid v_i\in I\} \cup B$, where $B$ is the set
  of blockers in $P$. Since $|B|=|E(G)|$, we have $|S|=k + |E(G)|$.
  Suppose towards a contradiction that~$S$ is not in general position,
  and let $q, r, s$ be three distinct collinear points in~$S$.
  By \cref{cond:bbp,cond:bbb}, and since
  the points in~$C$ are in a strictly convex position, it follows that
  exactly two of the points $q, r, s$ must be in~$C$.
  Suppose, without loss of generality, that $q=p_i, r=p_j \in C$ and
  $s \in B$. By \cref{cond:bpp}, there is an edge
  between the vertices~$v_i$ and~$v_j$ in~$G$ that correspond to the
  points $p_i, p_j \in C$, contradicting that $v_i, v_j \in I$.
  It follows that~$S$ is a subset in general position of
  cardinality~$k+|E(G)|$.

  Conversely, assume that~$S \subseteq P$ is in general position and
  that $|S|= k + |E(G)|$.
  By \Cref{obs:solutionstructure}, we may assume that $B \subseteq S$.
  Let~$I$ be the set of vertices corresponding to the points in~$S
  \setminus B$, and note that $|I|=k$. Since $B \subseteq S$, no two
  points $v_i, v_j$ in~$I$ can be adjacent;
  otherwise, their corresponding points $p_i, p_j$ and the blocker of
  edge $\{v_i,v_j\}$ would be three collinear points in~$S$.
  It follows that~$I$ is an independent set of cardinality~$k$ in~$G$.
\end{proof}
}

\Cref{lem:IS-to-GPS} implies the NP-hardness of \GPS.
Furthermore, a careful analysis of the proof of \Cref{lem:IS-to-GPS} reveals the intractability of an extension variant of \GPS, where as an additional input to the problem a subset~$S' \subseteq P$ of points in general position is given and the task is to find~$k$ \emph{additional} points in general position, that is, one looks for a point subset~$S \subseteq P$ in general position such that~$S' \subset S$ and~$|S| \ge |S'| + k$.
By \Cref{obs:solutionstructure}, we can assume for the instance created by transformation~$\Phi$ that~$B$ (the set of blockers) is contained in a maximum-cardinality point subset in general position. Thus,
we can set~$S' := B$.
The proof of \Cref{lem:IS-to-GPS} then shows that~$k$ points can be added to~$S'$ if and only if the graph~$G$ contains an independent set of size~$k$.
Since \textsc{Independent Set} is W[1]-hard with respect to the solution size~\cite{DF13},
we can observe the following:

\begin{observation}
  The extension variant of \GPS described above is W[1]-hard when parameterized by the number~$k$ of additional points. 
\end{observation}

Hence, this extension is \emph{not} fixed-parameter tractable with respect to~$k$, unless~W[1]${}={}$FPT.
The reader may want to contrast the W[1]-hardness result for the aforementioned extension variant of \GPS with the fixed-parameter tractability results for \GPS shown in \Cref{subsec:gps}.

Next, we turn our attention to approximation.
A closer inspection of transformation~$\Phi$ reveals that we can obtain a \PTAS-reduction from {\sc Maximum Independent Set} to the optimization version of \GPS.

\begin{definition}
  \label{def:PTAS-reduction}
  Given two maximization problems
  $Q$ and $Q'$, a {\em PTAS-reduction} from $Q$ to
  $Q'$ consists of three polynomial-time computable functions
  $f$, $f'$ and $\alpha\colon\mathds{Q}\to(1,\infty)$
  such that:
  \begin{compactenum}[(1)]
  \item For any instance $I$ of~$Q$ and for any constant $\rho>1$, $f$ produces an instance $I' = f(I, \rho)$ of $Q'$.
  \item For any solution $x'$ of $I'$ and for any $\rho>1$, $f'$
    produces a solution $x =f'(I,x',\rho)$ of~$I$ such that:
    $$\frac{\opt(I')}{|x'|} \leq \alpha(\rho)\; \Rightarrow\; \frac{\opt(I)}{|x|} \leq \rho.$$
  \end{compactenum}
\end{definition}

By {\sc IS-3} we denote the {\sc Maximum Independent Set} problem
restricted to graphs of maximum degree at most~3.
By \optGPS we denote the optimization version of \GPS
in which one seeks to compute a largest subset of points in general
position in a given point set.

\begin{lemma}\label[lemma]{lem:ptasreduction}
  There is a PTAS-reduction from {\sc IS-3} to {\rm \optGPS}.
\end{lemma}
{\begin{proof}
  Let $G$ be an instance of {\sc IS-3}, and note that $|E(G)| \leq
  3|V(G)|/2$. It is easy to see that~$G$ has an independent set
  of cardinality at least $|V(G)|/4$ that can be obtained by repeatedly
  selecting a vertex in~$G$ of minimum degree and discarding all its
  neighbors until the graph is empty.
  We define the computable functions $f$, $f'$ and $\alpha$ in
  \Cref{def:PTAS-reduction} as follows.
  The function~$f$, on input $(G,\rho)$, outputs $P:=\Phi(G)$;
  by \Cref{prop:valid}, $f$ is computable in polynomial time.
  Let $S$ be a subset in general position in~$P$.
  By \Cref{obs:solutionstructure}, there is a subset
  in general position $S'$ of cardinality at least $|S|$ that
  contains~$B$.
  We may assume that $|S'| \geq |B| + |V(G)|/4$; this assumption is
  justified because $G$ has an independent set of cardinality at least
  $|V(G)|/4$, and hence, by the proof of \Cref{lem:IS-to-GPS}, $P$ has a subset in
  general position of cardinality at least $|B|+|V(G)|/4$,
  which we may assume to contain $B$ by \Cref{obs:solutionstructure}.
  By the proof of \Cref{lem:IS-to-GPS}, $G$ has an independent set $I$ of
  cardinality $|S'| - |B| \geq |V(G)|/4$. We define $f'(G, S', \rho) := I$,
  which is clearly polynomial-time computable. Finally, we define $\alpha(\rho) := (\rho + 6)/7$.

  Let~$\opt(G)$ denote the cardinality of a maximum independent set
  in~$G$, and let $\opt(P)$ be the cardinality of a largest subset in
  general position in~$P$.
  From \Cref{lem:IS-to-GPS}, it follows that $\opt(P) = |B| +
  \opt(G)$. Let $S$ be an approximate solution to $P$,
  and by the discussion above, we may assume that $S$ contains~$B$ and
  is of cardinality at least $|B|+|V(G)|/4$. Let $I =f'(G, S, \rho)$, and note
  that $|I| \geq |V(G)|/4$. To finish the proof, we need to show that if
  $\opt(P)/|S| \leq (\rho +6)/7$, then $\opt(G)/|I| \leq \rho$.
  In effect, after noting that $|B|=|E(G)|\leq 3|V(G)|/2$ and $|I| \geq |V(G)|/4$, we have:
  \begin{alignat*}{3}
    & &\frac{\opt(P)}{|S|} & \leq \frac{\rho+6}{7}\\
    \iff & &\;\frac{|B|+\opt(G)}{|B|+|I|} & \leq \frac{\rho+6}{7}\\
    \iff & &\;\frac{\opt(G)}{|I|} & \leq \frac{\rho+6}{7} \cdot \frac{|B|+|I|}{|I|}-\frac{|B|}{|I|}\\
    & & &= \frac{\rho+6}{7} \cdot \left(\frac{|B|}{|I|} + 1\right)-\frac{|B|}{|I|}\\
    & & &= \frac{\rho-1}{7} \cdot \frac{|B|}{|I|} + \frac{\rho+6}{7}\\
    & & &\leq \frac{\rho-1}{7}\cdot 6  + \frac{\rho+6}{7}\\
    & & &=\rho.
  \end{alignat*}
\end{proof}}

Finally, we prove that transformation~$\Phi$ yields a polynomial-time reduction from
{\sc IS-3} to \optGPS, where the number of points in the point
set depends linearly on the number of vertices in the graph.
This implies an exponential-time lower bound based on
the Exponential Time Hypothesis (ETH)~\cite{IP99}.

\begin{lemma}\label[lemma]{lem:serfreduction}
  There is a polynomial-time reduction from {\sc IS-3}
  to \GPS mapping a graph~$G$ to a point set~$P$ of
  size~$O(|V(G)|)$.
\end{lemma}
{\begin{proof}
  For an instance $G$ of {\sc IS-3}, the set $P := \Phi(G)$ is
  of cardinality $|P|=|V(G)|+|E(G)|\le |V(G)|+3|V(G)|/2 \in O(|V(G)|)$.
  By the proof of \Cref{lem:IS-to-GPS}, $\Phi$ is a polynomial-time reduction.
\end{proof}}

We summarize the consequences of \Cref{lem:IS-to-GPS,lem:ptasreduction,lem:serfreduction} in the following theorem:
}

\begin{theorem}
  \label{thm:NP_APX_hard}
The following are true:
  \begin{compactenum}[(a)]
    \item\label{NP} \GPS is \NP-complete.
    \item\label{APX} {\rm \optGPS} is \APX-hard.
    \item\label{ETH} Unless ETH fails, \GPS is not solvable in~$2^{o(n)}\cdot n^{O(1)}$ time. 
  \end{compactenum}

We note that the theorem even holds for the restriction of \GPS to instances in which no four points are collinear.
\end{theorem}
\lv{\begin{proof}
  Part~(\ref{NP}) follows from the \NP-hardness of {\sc Independent Set}~\cite{GJ79}, combined with \Cref{prop:valid} and \Cref{lem:IS-to-GPS} (membership in \NP~trivially holds).
  Part~(\ref{APX}) follows from the \APX-hardness of {\sc
    IS-3}~\cite{AK00}, combined with \Cref{prop:valid} and \Cref{lem:ptasreduction}.
  Concerning Part~(\ref{ETH}), it is well known that, unless ETH fails, {\sc Maximum Independent
    Set} is not solvable in subexponential time~\cite{IPZ01}, and the same is true for {\sc
    IS-3} by the results of \citet{JohnsonSzegedy99}.
  Hence, by the reduction in~\Cref{lem:serfreduction}, \GPS
  cannot be solved in subexponential time since this would imply
  a subexponential-time algorithm for {\sc IS-3}.

  Parts (a)--(c) remain true for the restriction of \GPS to instances in which no four points are collinear because the point set produced by transformation~$\Phi$ satisfies this property (see \Cref{prop:valid}).
\end{proof}}

Currently, the best approximation result for \optGPS is due to \citet{Cao12}, who provided a simple greedy $\sqrt{\opt}$-factor approximation algorithm.
Therefore, a large gap remains between the proven upper and the lower bound on the approximation factor.

\section{Fixed-Parameter Tractability}
\label{sec:fpt}
In this section, we prove several fixed-parameter tractability results for
\GPS.
In \Cref{subsec:gps} we develop cubic-size problem kernels with
respect to the parameter size~$k$ of the sought subset
in general position, and with respect to the line cover number~$\ell$.
In \Cref{subsec:dual}, we show a quadratic-size problem kernel with respect
to the \emph{dual parameter}~$h:=n-k$, that is, the number of points whose
deletion leaves a set of points in general position. Moreover,
we prove that this problem kernel is essentially optimal, unless an unlikely collapse in the polynomial hierarchy occurs.

\subsection{Fixed-Parameter Tractability Results for the Parameter Solution Size~$k$}
\label{subsec:gps}
Let $(P, k)$ be an instance of \GPS, and let~$n=|P|$.  \citet{Cao12} gave a problem kernel for \GPS of size~$O(k^4)$ based on the following idea.
Suppose that there is a line $L$ containing at least $\binom{k-2}{2}+2$ points from $P$. For any subset $S' \subset P$ in general position with~$|S'|=k-2$, there can be at most
 $\binom{k-2}{2}$ points on~$L$ such that each is collinear with two points in~$S'$. Hence, we can always find at least two points on $L$ that together with the points in $S'$ form a subset $S$
 in general position of cardinality $k$. Based on this idea, \citet{Cao12} introduced the following data reduction rule:

\begin{rrule}[\cite{Cao12}]
    \label[rrule]{rul:heavy_line}
	Let~$(P,k)$ be an instance of \GPS.
	If there is a line $L$ that contains at least~$\binom{k-2}{2}+2$ points from $P$,
        then remove all the points on $L$ and set~$k:=k-2$.
\end{rrule}

Cao showed that \Cref{rul:heavy_line} can be exhaustively
applied in $O(n^3)$ time (\cite[Lemma~B.1.]{Cao12}), and he showed its correctness, that is, applying~\Cref{rul:heavy_line} to an instance~$(P,k)$ yields an instance~$(P',k')$ which is a yes-instance if and only if $(P,k)$ is (\cite[Theorem~B.2.]{Cao12}).
Using \Cref{rul:heavy_line}, he gave a kernel for \GPS of size~$O(k^4)$ that is computable in $O(n^3)$ time (\cite[Theorem~B.3.]{Cao12}).
We shall improve on \citeauthor{Cao12}'s result, both in terms of the kernel size and the running time of the kernelization algorithm.
We start by showing how, using a result by \citet[Theorem~3.2]{guibasetal96}, \Cref{rul:heavy_line} can be applied exhaustively in $O(n^2\log{n})$ time.
Notably, the idea of reducing lines with many points (based on \citet{guibasetal96}) also yields kernelization results for \textsc{Point Line Cover} \cite{LM05}.

 \begin{lemma}
 \label[lemma]{lem:heavylinetime}
 Given an instance $(P, k)$ of \GPS where $|P|=n$, in $O(n^2\log{n})$ time we can compute an equivalent instance $(P', k')$ such that either $(P', k')$ is a trivial yes-instance, or
 the collinearity of $P'$ is at most~$\binom{k-2}{2}+1$.
 \end{lemma}

 {\begin{proof}
	Let $\lambda=\binom{k-2}{2}+2$. 
	We start by computing the set ${\cal L}$ of all lines that contain at least $\lambda$ points from $P$.
	By a result of \citet[Theorem~3.2]{guibasetal96}, this can be performed in $O(n^2\log{(n/\lambda)/\lambda)}$ time (the algorithm also yields for every such line the points of~$P$ lying on that line).
	We then iterate over each line~$L\in{\cal L}$, checking whether $L$, at the current iteration, still contains at least $\lambda$ points; if it does, we remove all points on $L$ from $P$ and decrement $k$ by 2.
	For each line~$L$, the running time of the preceding step is $O(\lambda)$, which is the time to check whether~$L$ contains at least~$\lambda$ points. Additionally, we might need to remove all points on~$L$.
	If~$k$ reaches zero, we can return a trivial yes-instance $(P', k')$ of \GPS in constant time.
	Otherwise, after iterating over all lines in ${\cal L}$, by \Cref{rul:heavy_line}, the resulting instance $(P', k')$ is an equivalent instance to $(P, k)$ satisfying that no line in $P'$ contains $\lambda$ points, and hence the collinearity of $P'$ is at most $\binom{k-2}{2}+1$.
	Overall, the above can be implemented in time $O((n^2\log{(n/\lambda)/\lambda)\cdot \lambda)=O(n^2\log{n})}$.
 \end{proof}}

 We move on to improving the size of the problem kernel.
\citet[Theorem~2.3]{PW13} proved a lower bound on the maximum cardinality of
a subset in general position when an upper bound on the collinearity of the point set is known. We show next how to obtain a kernel for \GPS of cubic size based on this result of \citet{PW13}.

\begin{theorem}\label{thm:cubic-kernel}
  \GPS admits a problem kernel containing at most $15k^3$ points
  that is computable in~$O(n^2\log{n})$ time.
\end{theorem}

{\begin{proof}
	By \Cref{lem:heavylinetime}, after $O(n^2\log{n})$ preprocessing time, we can either return an equivalent yes-instance of $(P, k)$ of constant size, or obtain an equivalent instance for which the collinearity of the point set is at most $\binom{k-2}{2}+1$.
	Therefore, without loss of generality, we can assume in what follows that the collinearity of $P$ is at most $\lambda=\binom{k-2}{2}+1$.

  \citet[Theorem~2.3]{PW13} showed that any set of~$n$~points whose collinearity is at
  most~$\lambda$ contains a subset of points in general position of size
  at least $\alpha n/\sqrt{n\ln{\lambda} + \lambda^2}$, for some constant $\alpha\in \R$.
  A lower bound of $\alpha \ge \sqrt{6}/72$ can be computed based on \citet[Lemmas 4.1, 4.2, and Theorems 2.2, 2.3, 4.3]{Payne14}.\lv{\footnote{From Theorems 2.2 and 2.3 we can deduce that Lemma 4.1 holds for the constant~$c=128$. Plugging~$c=128$ into the proof of Theorem 4.3 gives the desired lower bound~$\sqrt{6}/72$ for $\alpha$.}}
  Since $\lambda \leq \binom{k-2}{2}+1$, we can compute a value
  of~$n$, as a function of~$k$, above which we are guaranteed to have a
  subset in general position of cardinality at least~$k$.
  We do this by solving for~$n$ in the inequality
  $\alpha n/\sqrt{n\ln{\lambda} + \lambda^2} \geq k$ after substituting
  $\lambda$ with $\binom{k-2}{2}+1$ and $\alpha$ with $\sqrt{6}/72$.
  We obtain that if $n \geq 15k^3$, then the aforementioned inequality is satisfied for all $k \geq 29337$.
	The kernelization algorithm distinguishes the following three cases:
	First, if $k < 29337$, then the algorithm decides the instance in~$O(1)$ time, and returns an equivalent instance of~$O(1)$ size.
	Second, if $k \ge 29337$ and $n \geq 15k^3$, then the algorithm returns a trivial yes-instance of constant size.
	Third, if none of the two above cases applies, then it returns the (preprocessed) instance $(P, k)$ which satisfies $|P| \leq 15k^3$.
\end{proof}}

We can derive the following result by a brute-force algorithm on the above problem kernel:
\begin{corollary}
\label[corollary]{cor:fptalgo}
\GPS can be solved in $O(n^2\log{n}+41^k \cdot k^{2k})$ time.
\end{corollary}

\begin{proof}
  Let $(P, k)$ be an instance of \GPS. By \Cref{thm:cubic-kernel}, after $O(n^2\log{n})$ preprocessing time, we can assume that $|P| \leq 15k^3$.
  We enumerate every subset of size $k$ in $P$, and for each such subset, we use the result of~\citet[Theorem~3.2]{guibasetal96} to check in $O(k^2\log{k})$ time whether the subset is in general position.
  If we find such a subset, then we answer positively; otherwise (no such subset exists), we answer negatively.
  The number of enumerated subsets is
  \begin{align*}\binom{|P|}{k} &\leq\binom{15k^3}{k} \leq\frac{(15k^3)^k}{k!}\\
    &\leq \frac{(15k^3)^k}{(k/e)^k} = (15ek^3/k)^k \leq (40.78)^k k^{2k},\end{align*}
	where $e$ is the base of the natural logarithm and~$k! \ge (k/e)^k$ follows from Stirling's formula.
  Putting everything together, we obtain an algorithm for \GPS that
  runs in $O(n^2\log{n}+ (40.78)^{k} \cdot k^{2k} \cdot k^2\log{k}) =
  O(n^2\log{n}+ 41^{k} \cdot k^{2k})$ time.
\end{proof}

Let 3-\GPS denote the restriction of \GPS to instances in which the point set contains no four collinear points. By \Cref{thm:NP_APX_hard}, 3-\GPS is \NP-complete.
\citet[Theorem~1]{furedi91} showed that every set $P$ of $n$ points in which no four points are collinear contains a subset in general position of size $\Omega(\sqrt{n \log{n}})$.
Based on F{\"{u}}redi's result and the idea in the proof of \Cref{thm:cubic-kernel}, we get:

\begin{corollary}
3-\GPS admits a problem kernel containing $O(k^2/\log{k})$ points
that is computable in~$O(n)$ time.
\end{corollary}

\citet{Cao12} made the following observation on the relation between the cardinality of a maximum-cardinality point subset in general position and the \emph{line cover number}, that is, the minimum number of lines that cover all points in the point set.
For the sake of self-containment, we also give a short proof.

\begin{observation}[\cite{Cao12}]
  \label[observation]{obs:line}
  For a set~$P$ of points let $S\subseteq P$ be a maximum
  subset in general position and let~$\ell$ be the line cover
  number of~$P$.
  Then, $\sqrt{\ell} \le |S| \le 2\ell$.
\end{observation}

\begin{proof}
   For the first inequality, note that~$|S|$ points in general position
   define $\binom{|S|}{2}\le |S|^2$ lines. Since all other points
   in~$P$ have to lie on a line defined by two points in~$S$, it
   follows that~$\ell\le |S|^2$.
   The second inequality clearly holds since any maximum subset in
   general position can contain at most two points that lie on the same line.
 \end{proof}

As a consequence of \Cref{obs:line}, we can assume that $k \le 2\ell$ and, thus, we can transfer our results for the parameter~$k$ to the parameter~$\ell$.

\begin{corollary}
  \label[corollary]{cor:line_cover}
  \GPS can be solved in $O(n^2\log{n}+41^{2\ell} \cdot (2\ell)^{4\ell})$ time, and there is a kernelization algorithm that, given an instance $(P, k)$ of \GPS, computes an equivalent instance containing at most $120\ell^3$ points in~$O(n^2\log{n})$ time.
\end{corollary}

\subsection{Fixed-Parameter Tractability Results for the Dual Parameter~$h$}
\label{subsec:dual}
In this section we consider the dual parameter number
$h:=n-k$ of points that have to be {\em deleted} (\ie, excluded from the sought point set in general position) so that the
 remaining points are in general position.
We show a problem kernel containing $O(h^2)$ points for \GPS.
Moreover, we show that most likely this problem kernel is essentially
\emph{tight}, that is, there is
presumably no problem kernel with~$O(h^{2-\epsilon})$ points for any~$\epsilon > 0$.

We start with the problem kernel that relies essentially on a problem
kernel for the \HS{3} problem:

\problemdef{\HS{3}}
{A universe~$U$, a collection~$\mathcal{C}$ of size-3 subsets of~$U$, and $h\in\N$.}
{Is there a subset $H\subseteq U$ of size at most~$h$ containing at
  least one element from each subset~$S\in\mathcal{C}$?}

\noindent There is a close connection between \GPS and
\HS{3}: For any collinear triple $p,q,r \in P$ of
distinct points, one of the three points has to be deleted in order to obtain a
subset in general position.
Hence, the set of deleted points has to be a  hitting set for the family of all collinear triples in~$P$.
Since \HS{3} can be solved in~$O(2.08^h + |\mathcal{C}|+|U|)$ time~\cite{Wah07}, we get:

\begin{proposition}\label[proposition]{prop:fptalgoDualParam}
	\GPS can be solved in $O(2.08^h + n^{3})$ time.
\end{proposition}

\HS{3} is known to admit a problem kernel with a universe of size~$O(h^2)$ computable in~$O(|U|+|\mathcal{C}| + h^{1.5})$ time~\cite{Bev14}.
Based on this, one can obtain a problem kernel of size~$O(h^2)$ computable in~$O(n^3)$ time.
The bottleneck in this running time is listing all collinear triples.
We can improve the running time of this kernelization algorithm by giving a direct kernel exploiting the simple geometric fact that two non-parallel lines intersect in one point.
We first need two reduction rules for which we introduce the following definition:

\begin{definition}
  For a set~$P$ of points in the plane, we say that a point~$p\in P$ is in \emph{conflict} with a point~$q\in P$ if there is a third point~$z\in P$ such that~$p$, $q$ and~$z$ lie on the same line.
\end{definition}

\begin{rrule}
    \label[rrule]{rul:non_conflict}
	Let~$(P,k)$ be an instance of \GPS.
	If there is a point~$p\in P$ that is not in conflict with any other points in~$P$, then delete~$p$ and decrease~$k$ by one.
\end{rrule}

Clearly, \Cref{rul:non_conflict} is correct since we can always add a point which is not lying on any line defined by two other points to a general position subset.
The next rule deals with points that are in conflict with too many other points.
The basic idea here is that if a point lies on more than~$h$ distinct lines defined by two other points of~$P$, then it has to be deleted.
This is generalized in the next rule.
\begin{rrule}
    \label[rrule]{rul:conflicts}
	Let~$(P,k)$ be an instance of \GPS.
	For a point~$p\in P$, let~$\mathcal{L}(p)$ be the set of lines containing~$p$ and at least two points of~$P \setminus \{p\}$, and for~$L \in \mathcal{L}(p)$ let~$|L|$ denote the number of points of~$P$ on~$L$.
	Then, delete each point~$p\in P$ satisfying $\sum_{L\in\mathcal{L}(p)} (|L| - 2) > h$.
\end{rrule}

\begin{lemma}
  \label[lemma]{lem:conflicts_correct}
  Let~$(P,k)$ be a \GPS instance and let~$(P',k)$ be the resulting instance after applying \Cref{rul:conflicts} to~$(P,k)$.
  Then, $(P,k)$ is a yes-instance if and only if~$(P',k)$ is a yes-instance.
\end{lemma}
\begin{proof}
  Let~$(P,k)$ be an instance of \GPS and let~$(P':=P\setminus D,k)$ be the reduced instance, where~$D\subseteq P$ denotes the set of removed points.

  Clearly, if~$(P',k)$ is a yes-instance, then so is~$(P,k)$.
  For the converse, we show that any size-$k$ subset of~$P$ in general position does not contain any point~$p\in D$:
  For each line~$L\in\mathcal{L}(p)$, all but two points need to be deleted.
  If a subset~$S \subseteq P$ in general position contains~$p$, then the points that have to be deleted on the lines in~$\mathcal{L}(p)$ are all distinct since any two of these lines only intersect in~$p$.
  This means that~$\sum_{L\in\mathcal{L}(p)} (|L| - 2)$ points need to be deleted.
  However, since this value is by assumption larger than~$h$, the solution~$S$ is of size less than~$k = |P| - h$.
\end{proof}

\begin{theorem}
  \label{thm:dual_kernel}
  \GPS admits a problem kernel containing at most $2h^2+h$ points that is computable in~$O(n^2)$ time, where~$h=n-k$.
\end{theorem}

\begin{proof}
  Let~$(P,k)$ be a \GPS instance.
  We first show that applying \Cref{rul:non_conflict} exhaustively and then applying \Cref{rul:conflicts} once indeed gives a small instance~$(P',k')$.
  Note that each point~$p\in P'$ is in conflict with at least two other points, that is, $p$~is on at least one line containing two other points in~$P'$, since the instance is reduced with respect to~\Cref{rul:non_conflict}.
  Moreover, since the instance is reduced with respect to~\Cref{rul:conflicts}, it follows that each point is in conflict with at most~$2h$ other points.
  Thus, deleting~$h$ points can give at most~$h \cdot 2h$ points in general position.
  Hence, if~$P'$ contains more than~$2h^2+h$ points, then the input instance is a no-instance.

  We next show how to apply \Cref{rul:non_conflict,rul:conflicts} in~$O(n^2)$ time.
  To this end, we follow an approach described by \citet{EOS86} and \citet{GRT97} which uses the dual representation and line arrangements.
  The dual representation maps points to lines as follows:~$(a,b) \mapsto y = ax + b$.
  A line in the primal representation containing some points of~$P$ corresponds in the dual representation to the intersection point of the lines corresponding to these points.
  Thus, a set of at least three collinear points in the primal corresponds to the intersection of the corresponding lines in the dual.
  An \emph{arrangement} of lines in the plane is, roughly speaking, the partition of the plane formed by these lines.
  A representation of an arrangement of $n$ lines can be computed in~$O(n^2)$ time~\cite{EOS86}.
  Using the algorithm of \citet{EOS86}, we compute in~$O(n^2)$ time the arrangement~$A(P^*)$ of the lines~$P^*$ in the dual representation of~$P$.

  \Cref{rul:non_conflict} is now easily computable in~$O(n^2)$ time:
  Initially, mark all points in~$P$ as ``not in conflict''.
  Then, iterate over the vertices of~$A(P^*)$ and whenever the vertex has degree six or more (each line on an intersection contributes two to the degree of the corresponding vertex) mark the points corresponding to the intersecting lines as ``in conflict''.
  In a last step, remove all points that are marked as ``not in conflict''.

  \Cref{rul:conflicts} can be applied in a similar fashion in~$O(n^2)$ time:
  Assign a counter to each point~$p\in P$ and initialize it with zero.
  We want this counter to store the number $\sum_{L\in\mathcal{L}(p)} (|L| - 2)$ on which \Cref{rul:conflicts} is conditioned.
  To this end, we iterate over the vertices in~$A(P^*)$ and for each vertex of degree six or more (each line contributes two to the degree of the intersection vertex) we increase the counter of each point corresponding to a line in the intersection by $d/2 - 2$ where~$d$ is the degree of the vertex.
  After one pass over all vertices in~$A(P^*)$ in~$O(n^2)$ time, the counters of the points store the correct values and we can delete all points whose counter is more than~$h$.
\end{proof}

We remark that the results in \Cref{prop:fptalgoDualParam} and \Cref{thm:dual_kernel} also hold if we replace the parameter~$h$ by the ``number~$\gamma$ of \emph{inner} points'', where we call a point an inner point if it is not a vertex of the convex hull of~$P$.
The reason is that in all non-trivial instances we have~$h \le \gamma$ since removing all inner points yields a set of points in general position.

We can prove a matching (conditional) lower bound on the problem kernel size for \GPS via a reduction from \VC.
Given an undirected graph~$G$ and $k\in\N$, \VC asks whether there is a subset~$C$ of at most~$k$ vertices such that every edge is incident to at least one vertex in~$C$.
Using a lower bound result by~\citet{DM14} for \VC (which is based on the common assumption in complexity theory that coNP is not in NP/poly since otherwise the polynomial hierarchy collapses to its third level), we obtain the following:

\begin{theorem}
  \label{thm:dual_size_lower_bound}
  Unless $\coNPinNPpoly$,  for any
  $\epsilon > 0$, \GPS admits no problem kernel of size~$O(h^{2-\epsilon})$.
\end{theorem}
\begin{proof}
  We give a polynomial-time reduction from \VC, where the resulting dual parameter~$h$
  equals the size of the sought vertex cover.
  The claimed lower bound then follows because, unless
  $\coNPinNPpoly$, for any~$\epsilon > 0$,
   \VC admits no problem kernel of
  size~$O(k^{2-\epsilon})$, where~$k$ is the size of the vertex cover~\cite{DM14}.

  Given a \VC instance~$(G,k)$,
  we first reduce it to the equivalent \IS instance~$(G,|V(G)|-k)$.
  We then apply transformation~$\Phi$ (see \Cref{sec:hardness}) to $G$ to
  obtain a set of points~$P$, where~$|P|=|V(G)|+ |E(G)|$; we set
  $k':=|V(G)|+|E(G)|-k$, and consider the instance $(P, k')$ of \GPS.
  Clearly, $G$ has a vertex cover of cardinality~$k$ if and only if~$G$ has an
  independent set of cardinality~$|V(G)|-k$, which\lv{, by \Cref{lem:IS-to-GPS},}
  is true if and only if~$P$ has a subset in general position of
  cardinality $|E(G)| + |V(G)| - k$. 
  Hence, the dual parameter~$h=|P|-k'$ equals the sought vertex cover size.
\end{proof}

Note that~\Cref{thm:dual_size_lower_bound} gives a lower
bound only on the \emph{total size} (\ie, instance size) of a problem kernel for \GPS.
We can show a stronger lower bound on the number of points
contained in any problem kernel using ideas from~\citet{KPR16}, which are based on
a lower bound framework by \citet{DM14}.
\citet{KPR16} show that there is no polynomial-time algorithm that reduces a
\PLC instance~$(P,k)$ to an equivalent instance with~$O(k^{2-\epsilon})$ points for any~$\epsilon>0$ unless $\coNPinNPpoly$.
The proof is based on a result by \citet{DM14}
who showed that \VC does not admit a so-called \emph{oracle communication protocol}
of cost~$O(k^{2-\epsilon})$ for~$\epsilon > 0$ unless $\coNPinNPpoly$.
An oracle communication protocol is a two-player protocol, in which one player is holding the input and is allowed polynomial (computational) time
in the length of the input, and the second player is
computationally unbounded. The \emph{cost} of the communication
protocol is the number of bits communicated from the first player to
the second player in order to solve the input instance.

\citet{KPR16} devise an oracle communication protocol
of cost~$O(n\log n)$ for deciding instances of \PLC with~$n$
points. Thus, a problem kernel for \PLC with~$O(k^{2-\epsilon})$ points
implies an oracle communication protocol of cost~$O(k^{2-\epsilon'})$ for some
$\epsilon' > 0$ since the first player could simply compute the kernelized
instance in polynomial time and subsequently apply the protocol
yielding a cost of $O(k^{2-\epsilon}\cdot\log (k^{2-\epsilon}))$, which is
in~$O(k^{2-\epsilon'})$ for some~$\epsilon' > 0$.
This again implies an $O(k^{2-\epsilon''})$-cost oracle communication
protocol for \VC for some~$\epsilon''>0$ (via a polynomial-time reduction with a linear
parameter increase~\cite[Lemma~6]{KPR16}).
We show that there exists a similar oracle communication
protocol of cost~$O(n\log n)$ for \GPS.

The protocol is based on \emph{order types} of point sets.
Let~$P=\langle p_1,\ldots,p_n\rangle$ be an ordered set of points and
denote by~$\binom{[n]}{3}$ the set of ordered triples~$\langle i, j,
k\rangle$ where~$i < j < k$, $i,j,k\in [n]:=\{1,\ldots,n\}$.
The \emph{order type} of~$P$ is a function~$\sigma:
\binom{[n]}{3}\rightarrow\{-1,0,1\}$, where~$\sigma(\langle
i,j,k\rangle)$ equals~$1$ if $p_i$, $p_j$, $p_k$ are in
counter-clockwise order, equals~$-1$ if they are in clockwise order,
and equals~$0$ if they are collinear. Two point sets~$P$ and~$Q$
of the same cardinality are \emph{combinatorially equivalent} if there exist
orderings~$P'$ and~$Q'$ of~$P$ and~$Q$ such that the order types
of~$P'$ and~$Q'$ are identical.

A key step in the development of an oracle communication protocol
is to show that two instances of \PLC with
combinatorially equivalent point sets are actually equivalent~\cite[Lemma~2]{KPR16}.
We can prove an analogous result for \GPS:

\begin{observation}
  \label[observation]{lem:combinatorial_equivalent}
  Let~$(P,k)$ and~$(Q,k)$ be two instances of \GPS. If the point
  sets~$P$ and~$Q$ are combinatorially equivalent, then~$(P,k)$
  and~$(Q,k)$ are equivalent instances of \GPS.
\end{observation}
\begin{proof}
  Let~$P$ and~$Q$ be combinatorially equivalent point sets
  with~$|P|=|Q|=n$ and let~$P'=\langle p_1,\ldots,p_n \rangle$
  and~$Q'=\langle q_1,\ldots,q_n \rangle$ be orderings of~$P$ and~$Q$,
  respectively, having the same order type~$\sigma$.

  Now, a subset~$S\subseteq P'$ is in general position if and only if
  no three points in~$S$ are collinear, that is,
  $\sigma(\langle p_i, p_j, p_k\rangle)\neq 0$ holds for all~$p_i,
  p_j, p_k\in S$. Consequently, it holds that~$\sigma(\langle
  q_i,q_j,q_k\rangle)\neq 0$, and thus the subset~$\{q_i\mid p_i\in
  S\}\subseteq Q'$ is in general position.
  Hence,~$(P,k)$ is a yes-instance if and only if~$(Q,k)$ is a yes-instance.
\end{proof}

Based on \Cref{lem:combinatorial_equivalent}, we obtain an oracle
communication protocol for \GPS. The proof of the following lemma is completely analogous to the proof
of Lemma~4.1 in~\cite{KPR16}:

\begin{lemma}
  \label[lemma]{lem:oracle}
  There is an oracle communication protocol of cost~$O(n\log n)$ for
  deciding instances of \GPS with~$n$ points.
\end{lemma}

The basic idea is that the first player only sends the order type
of the input point set so that the computationally unbounded second player
can solve the instance (according to \Cref{lem:combinatorial_equivalent} the order type
contains enough information to solve a \GPS instance). We conclude with the following lower bound result:

\begin{theorem}
  \label{thm:dual_points_lower_bound}
  Let $\epsilon > 0$. Unless $\coNPinNPpoly$, there is no
  polynomial-time algorithm that reduces an instance~$(P,k)$ of \GPS
  to an equivalent instance with~$O(h^{2-\epsilon})$ points.
\end{theorem}
\begin{proof}
  Assuming that such an algorithm exists, the oracle communication protocol of
  \Cref{lem:oracle} has cost~$O(h^{2-\epsilon'})$ for some~$\epsilon'>0$.
  Since the reduction from \VC in \Cref{thm:dual_size_lower_bound}
  outputs a \GPS instance where the dual parameter~$h$ equals the size~$k$ of the vertex cover sought,
  we obtain a communication protocol for \VC~of cost~$O(k^{2-\epsilon'})$, which implies that
  $\coNPinNPpoly$ \cite[Theorem~2]{DM14}.
\end{proof}

\paragraph*{Remark on the Kernel Lower Bound Framework of Kratsch, Philip
  and~Ray}
As~a final observation, we mention that the framework of~\citet{KPR16} indeed is more generally applicable than stated there.
It only relies on the equivalence of instances with respect to order types of point
sets. Hence, we observe that
for every decision problem on point sets for which
\begin{compactenum}
  \item two instances with combinatorially equivalent point sets are
    equivalent (cf. \Cref{lem:combinatorial_equivalent}), and
  \item there is no oracle communication protocol of
      cost~$O(k^{2-\epsilon})$ for some parameter~$k$ and any~$\epsilon>0$
      unless $\coNPinNPpoly$,
\end{compactenum}
there is no problem kernel with~$O(k^{2-\epsilon'})$ points
for any~$\epsilon'>0$ unless $\coNPinNPpoly$.

\section{Conclusion and Outlook}

The intent of our work is to stimulate further research on the computational complexity of \GPS.
The kernelization results we presented rely mostly on combinatorial arguments; the main geometric property we used is
that two distinct lines intersect in at most one point.
Therefore, a natural question to ask is whether there are further geometric
properties that can be exploited in order to obtain improved algorithmic results for \GPS.
We conclude with the following concrete open questions:
\begin{enumerate}
	\item Can the~($15k^3$)-point kernel (\Cref{thm:cubic-kernel}) for \GPS be asymptotically improved?
		Or can we derive a cubic, or even a quadratic, lower bound on the (point) kernel size of \GPS?
	\item Can the FPT algorithm (see \Cref{cor:fptalgo}) for \GPS be (significantly) improved?
	\item With respect to polynomial-time approximation, we could only show the APX-hardness of \optGPS.
	It remains open whether \citeauthor{Cao12}'s $O(\sqrt{\opt})$-factor approximation can be improved.
\end{enumerate}

\bibliography{ref}
\bibliographystyle{abbrvnat}

\end{document}